\documentclass{eptcs}

\usepackage[utf8]{inputenc}

\usepackage{setspace}
\usepackage{mathdots}

\usepackage{scalerel}

\usepackage{fancybox}

\usepackage[ddmmyy,hhmmss]{datetime}
\usepackage{tikzpagenodes}

\definecolor{myred}{RGB}{165,41,33}
\definecolor{mygrey}{RGB}{130,130,130}

\definecolor{navy}{rgb}{0,0,0.5}

\usepackage[english]{babel}

\usepackage{hyperref}
\usepackage{xcolor}
\hypersetup{
    colorlinks = false,
    linkbordercolor = {white},
    citebordercolor = {white},
    urlbordercolor = {white},
}

\usepackage{amsthm}
\usepackage[normalem]{ulem}
\definecolor{refkey}{gray}{.5}
\definecolor{labelkey}{gray}{.5}

\newcommand{\nats}{\mathbb{N}}

\newcommand{\Pow}{{\mathscr{P}}}
\newcommand{\pow}[1]{\Pow({#1})}

\newcommand{\st}{\mid}
\renewcommand{\leq}{\leqslant}
\renewcommand{\geq}{\geqslant}

\newcommand{\disj}{\mathsf{Disj}}
\newcommand{\Disj}[2]{\disj\!\left(#1,#2\right)}

\newcommand{\eqnothing}{=\!\!\varnothing}
\newcommand{\neqnothing}{\neq\!\!\varnothing}
\newcommand{\eqqnothing}{=\varnothing}
\newcommand{\neqqnothing}{\neq\varnothing}

\usepackage[T1]{fontenc}
\usepackage{enumitem}
\usepackage{varwidth}
\usepackage{amsmath}
\usepackage{xspace}
\usepackage{amssymb}
\usepackage{mathtools}
\usepackage{mathrsfs}
\usepackage[ruled]{algorithm}
\usepackage[noend]{algpseudocode}

\makeatletter

\newcommand{\COMMENT}[1]{}

\newcommand{\defAsEq}{\mathrel{\overset{\makebox[0pt]{\mbox{\normalfont\tiny\sffamily Def}}}{:=}}}

\newcommand{\Quote}[1]{#1}

\newcommand{\Vars}[1]{\mathit{Vars}(#1)}

\newcommand{\MLS}{\textnormal{\textsf{MLS}}\xspace}
\newcommand{\MLSSP}{\textnormal{\textsf{MLSSP}}\xspace}
\newcommand{\MLSP}{\textnormal{\textsf{MLSP}}\xspace}
\newcommand{\MLSU}{\textnormal{\textsf{MLSU}}\xspace}
\newcommand{\BST}{\textnormal{\textsf{BST}}\xspace}
\newcommand{\BSTC}{\textnormal{\textsf{BST$\!\times$}}\xspace}
\newcommand{\BSTuC}{\textnormal{\textsf{BST\raisebox{.026cm}{$\otimes$}}}\xspace}

\newcommand{\MSTtheory}[1]{\mathsf{MST}\hbox{$(#1)$}}

\newcommand{\MLSS}{\textnormal{\textsf{MLSS}}\xspace}
\newcommand{\MLSC}{\textnormal{\textsf{MLS$\!\times$}}\xspace}
\newcommand{\MLSI}{\textnormal{\textsf{MLSI}}\xspace}
\newcommand{\MLSuC}{\textnormal{\textsf{MLS\raisebox{.026cm}{$\otimes$}}}\xspace}

\newcommand{\NP}{\textnormal{\textsf{NP}}\xspace}
\newcommand{\defAs}{\coloneqq}
\newcommand{\true}{\textsf{true}}

\newcommand{\Theory}[1]{\BST(#1)}

\newcommand{\BSTTh}[1]{\mathsf{BST}\hbox{$(#1)$}}
\newcommand{\MSTTh}[1]{\mathsf{MST}\hbox{$(#1)$}}

\newcommand{\dom}[1]{\mathsf{dom}(#1)}
\newcommand{\rk}[1]{\mathsf{rk}\left(#1\right)}

\newcommand{\BSTbb}{\textnormal{$\mathbb{BST}$}\xspace}

\newcommand{\MSTbb}{\textnormal{$\mathbb{MST}$}\xspace}

\newtheorem{theorem}{Theorem}[section]
\newtheorem{lemma}[theorem]{Lemma}
\newtheorem{example}[theorem]{Example}

\newtheorem{corollary}[theorem]{Corollary}

\newcommand{\Varsf}{\Vars{\varphi}}

\title{On the Convexity of a Fragment\\ of Pure Set Theory with Applications \\within a Nelson-Oppen Framework\thanks{We gratefully acknowledge partial support from project ``STORAGE---Universit\`{a} degli Studi di Catania, Piano della Ricerca 2020/2022, Linea di intervento 2''.}}
\def\titlerunning{On the Convexity of a Fragment of Pure Set Theory}
\def\authorrunning{D.\ Cantone, A.\ De Domenico, and P.\ Maugeri}

\author{Domenico Cantone \qquad\qquad Pietro Maugeri
\institute{Dept. of Mathematics and Computer Science\\
University of Catania\\
Catania, Italy}
\email{domenico.cantone@unict.it \qquad\qquad pietro.maugeri@unict.it}
\and
Andrea De Domenico
\institute{Scuola Superiore di Catania\\University of Catania\\
Catania, Italy}
\email{andrea.dedomenico@studium.unict.it}
}
\date{}

\begin{document}

\maketitle
\begin{abstract}
    The Satisfiability Modulo Theories (SMT) issue concerns the satisfiability of formulae from multiple background theories, usually expressed in the language of first-order predicate logic with equality. SMT solvers are often based on variants of the Nelson-Oppen combination method, a solver for the quantifier-free fragment of the combination of  theories with disjoint signatures, via cooperation among their decision procedures. When each of the theories to be combined by the Nelson-Oppen method is \emph{convex} (that is, any conjunction of its literals can imply a disjunction of equalities only when it implies at least one of the equalities) and decidable in polynomial time, the running time of the combination procedure is guaranteed to be polynomial in the size of the input formula. 
In this paper, we prove the convexity of a fragment of Zermelo-Fraenkel set theory, called Multi-Level Syllogistic, most of whose polynomially decidable fragments we have recently characterized.\\[.2cm]
\textbf{Keywords:} convex theories, satisfiability modulo theories,  
decision problem, computable set theory.

\end{abstract}
\section*{Introduction}

In the process of developing reliable and provably correct software, it is often necessary to express and then subsequently verify properties that belong to different logical languages. Thus, the correctness of a software system depends on being able to prove these conditions, expressed in distinct first-order signatures with equality. %
The search for a satisfying assignment of a given formula with respect to some background first-order theory is known as the SMT (Satisfiability Modulo Theories) problem.

SMT solvers \cite{BST10} are particularly useful tools for the automated verification of properties expressed with quantifier-free first-order formulae. Some theories usually integrated with common SMT solvers are the theory of arrays, of bit-vectors, of linear arithmetic, and the theory of uninterpreted functions.

Every background theory used in some SMT solver comes along with its own satisfiability procedure. The problem of modularly combining such special-purpose algorithms is highly non-trivial, since without the appropriate restrictions it is not even decidable \cite{BGNRZ06}. 

We will now briefly introduce some definitions to understand how to tackle this question and under which assumptions one can do it effectively.

A first-order quantifier-free theory $T$, identified with the set of its theorems, is \emph{stably infinite} if every formula $\varphi$ satisfiable in $T$ is satisfiable in an infinite model of $T$. Let $\Sigma_{1}$ and $\Sigma_2$ be signatures for a first-order language. A $(\Sigma_1 \cup \Sigma_2)$-formula $\varphi$ is \emph{pure} if every literal in $\varphi$ is a $\Sigma_1$-literal or a $\Sigma_2$-literal. It is easy to see that every quantifier-free $(\Sigma_1 \cup \Sigma_2)$-formula $\varphi$ can be \emph{purified}, yet maintaining satisfiability, by (i) substituting every \emph{impure} subterm of the form $f(t)$ with $f(x)$, where $x$ is a new variable, (ii) adding to $\varphi$ the conjunct $x = t$, and (iii) recursively purifying the term $t$, if needed. 

We say that two theories $T_1$ and $T_2$ over the signatures $\Sigma_1$ and  $\Sigma_2$, respectively, are \emph{disjoint} when  $\Sigma_1$ and  $\Sigma_2$ do not share any non-logical symbols.\footnote{Besides propositional connectives, logical symbols comprise equality.} The Nelson-Oppen \cite{NO79} procedure provides a method for combining decision procedures for disjoint, stably infinite theories $T_1$ and $T_2$ into one for $T_1 \oplus T_2$, namely the $(\Sigma_1 \cup \Sigma_2)$-theory defined as the deductive closure of the union of the theories $T_1$ and $T_2$.

A theory $T$ is \emph{convex} if for all conjunctions of literals $\varphi$ in $T$ and for all nonempty disjunctions $\bigvee_{i =1}^n x_i = y_i$ of equalities, $\varphi$ implies $\bigvee_{i =1}^n x_i = y_i$ in $T$ if and only if $\varphi$ implies $x_i = y_i$ in $T$ for some $i \in \{1,...,n\}$.

Examples of convex theories are the theory of Linear Rational Arithmetic $\mathsf{T_{LRA}}$ and the theory of list structure $\mathsf{T_{L}}$.
	
The non-logical symbols of the theory of $\mathsf{T_{LRA}}$ are $+$, $-$, $\leq$, $0$, $1$; following \cite[Chapter 3.4.2]{BradleyManna2007}, its axioms (universally quantified) are:
	\begin{align*}
	&x + 0 = x, 
	&&x + (-x) = 0, \\
	&(x+y)+z = x+(y+z), 
	&&x + y = y + x, \\
	&x \leq y \land y \leq x \rightarrow x = y, 
	&&x \leq y \lor y \leq x, \\
	&x \leq y \rightarrow x + z \leq y + z,
	&&x \leq y \land y \leq z \rightarrow x \leq z,\\
    &nx = 0 \rightarrow x = 0,
	&& (\exists y)\ x = ny \qquad \text{(for each positive integer $n$),}
	\end{align*}
where $nx$ stands for $\underbrace{x+\cdots+x}_{n \text{ times}}$.
After \cite{NO79}, the non-logical symbols of the theory of list structure $\mathsf{T_{L}}$ are \textsf{car}, \textsf{cdr}, \textsf{cons}, and \textsf{atom}, and its axioms are:
	\begin{align*}
	&\textsf{car}(\textsf{cons}(x,y)) = x, \\
	&\textsf{cdr}(\textsf{cons}(x,y)) = y, \\
	&\neg \textsf{atom}(x) \rightarrow \textsf{cons}(\textsf{car}(x), \textsf{cdr}(x)) = x, \\
	&\neg \textsf{atom}(\textsf{cons}(x, y)),
	\end{align*}
where (i) \textsf{cons} is a binary function, with $\textsf{cons}(x, y)$ representing the list constructed by prepending the object $x$ to the list $y$, (ii) \textsf{car} and \textsf{cdr} are unary functions, the left and right projections, respectively, and (iii) \textsf{atom} is true if and only if $x$ is a single-element list. 

Given two disjoint stable infinite theories $T_1$ and $T_2$, the Nelson-Oppen combination technique establishes the satisfiability of a conjunction of pure formulae $\varphi_1 \land \varphi_2$ (where $\varphi_i$ has signature $\Sigma_i$) in $T_1 \oplus T_2$ from the decision procedures for $\varphi_1$ and $\varphi_2$. The key idea is to propagate equalities $x = y$ to $\varphi_2$ whenever $T_1 \cup \varphi_1$ implies $x = y$, and conversely. This iterative process can be performed quickly in polynomial time, when the theories involved are convex. On the other hand, case-splitting would occur when dealing with non-convex theories, since only one of the equalities of the disjunct implied by $\varphi_i$ must be chosen at every step.

In \cite{MultiZarba02, Zarba02, Zarba05}, variants of the Nelson-Oppen method were used to combine theories involving sets/multisets of  urelements (i.e., objects with no internal structure) with the theory of integers and with the theory of cardinal numbers in presence of a cardinality operator. The SMT problem in the context of the theory of finite sets is considered in \cite{BBRT18}.

\smallskip

In this paper, we start an investigation for combining decidable fragments of pure Zermelo-Fraenkel set theory (in which sets are recursively built up from other sets) with other theories within the Nelson-Oppen framework. More specifically, our main result is that the theory Multi-Level Syllogistic (the basic language of computable set theory---\MLS for short) is convex and therefore its decision procedure (and those of its several polynomial fragments \cite{CantoneDMO21,CMO20}) can be efficiently combined with the decision procedures of other basic decidable theories, such as for instance the theory of lists and the theory of linear rational arithmetic, since set theory is plainly stably infinite.

\begin{center}---------------\end{center}

The paper is organized as follows. Section~\ref{SEC:MLSsemantics} introduces the syntax and semantics of the theory \MLS of our interest. Then, in Section~\ref{SEC:convex}, we prove the main result of the paper, namely that the theory \MLS is convex. We also review several fragments of \MLS endowed with polynomial-time decision procedures, since these inherit convexity from \MLS and are therefore particularly interesting for efficient combinations with other convex decidable theories. Subsequently, in Section~\ref{non-convexity}, we prove the non-convexity of various extensions of \MLS. Finally, in Section~\ref{SEC:conclusions}, we provide some closing remarks and plans for future research.

\section{Syntax and semantics of \MLS}\label{SEC:MLSsemantics}
Multi-Level Syllogistic (\MLS) is the quantifier-free propositional 
closure of atoms of the types:
\begin{equation}\label{complete MLS suite}
x = \varnothing, \quad x = y, \quad x \subseteq y, \quad x \in y, \quad  x = y \setminus z, \quad x = y \cup z, \quad x = y \cap z,
\end{equation}
where $x,y,z$ stand for set variables. We denote by $\Vars{\varphi}$ the collection of the set variables occurring in any \MLS-formula $\varphi$.

The satisfiability problem for \MLS has been first solved in the seminal paper \cite{FOS80b}. Its \NP-completeness (and that of its extension \MLSS with the singleton operator) has been later proved in \cite{CanOmoPol90}. Several extensions of \MLS have been proved decidable over the years, giving rise to the field of Computable Set Theory (see \cite{CanFerOmo89a,CanOmoPol01,SchCanOmo11,CanUrs18} for an in-depth account).

The semantics of \MLS is defined in the most natural way by means of \textit{set assignments}.

A \textsc{set assignment} $M$ is any map from a \emph{finite} collection of set variables $V$, %
denoted $\textsf{dom}(M)$,
into the von Neumann universe \textbf{$\mathcal{V}$}. 

We recall that \textbf{$\mathcal{V}$} is the cumulative hierarchy constructed in stages by transfinite recursion over the class $\mathit{On}$ of all ordinals. Specifically, 
\[
\textbf{$\mathcal{V}$} \coloneqq \bigcup_{\alpha \in \mathit{On}}\mathcal{V}_\alpha,
\]
where, recursively, 
\[
\mathcal{V}_{\alpha} \coloneqq \bigcup_{\beta < \alpha}\pow{\mathcal{V}_\beta},
\]
for every $\alpha \in \mathit{On}$, with $\pow{\cdot}$ denoting the powerset operator.

The notion of rank of a set is strictly connected to the construction steps of the von Neumann hierarchy. Specifically, for any set $s \in $\textbf{$\mathcal{V}$}, the rank of $s$ (denoted $\rk{s}$) is defined as the least ordinal $\alpha$ such that $s \subseteq \mathcal{V}_\alpha$. The rank function is extended to set assignments $M$, by putting $\rk{M} \defAs \max \{\rk{Mx} \st x \in \dom{M}\}$.

The set operators and relators of \MLS are interpreted according to their usual semantics. Thus, given a set assignment $M$, we put:
\begin{align*}
M(x \star y) &\coloneqq Mx \star My\\
\shortintertext{and}
M(x = y) = \true &\iff Mx = My, \\
M(x \in y) = \true &\iff Mx \in My, \\
M(x = y \star z) = \true &\iff Mx = M(y \star z),
\end{align*}
where $\star \in \{\cup,\,\cap,\, \setminus\}$ and $x,y,z \in \textsf{dom}(M)$.

Finally, for all $\MLS$-formulae $\varphi$ and $\psi$, we put by structural recursion:
\begin{align*}
&M(\neg\varphi) \coloneqq \neg M(\varphi), &M(\varphi \wedge \psi) \coloneqq M\varphi \wedge M\psi,\\
&M(\varphi \vee \psi) \coloneqq M\varphi \vee M\psi, &M(\varphi \rightarrow \psi) \coloneqq M\varphi \rightarrow M\psi.
\end{align*}
An \MLS-formula $\varphi$ is \textsc{satisfiable} if there exists a set assignment $M$ over $\Vars{\varphi}$ such that $M\varphi = \true$, in which case we also write $M \models \varphi$ and say that $M$ is a \textsc{model} for $\varphi$. 
If $\varphi$ is satisfied by all set assignments, we say that $\varphi$ is \textsc{true} and write $\models \varphi$.

By way of disjunctive normal form, the satisfiability problem for \MLS can be reduced to the satisfiability problem for conjunctions of \MLS-literals, namely \MLS-atoms of types \eqref{complete MLS suite} and their negation. 
In addition, for the purposes of simplifying some proofs, we can further restrict ourselves to \MLS-conjunctions involving a minimal number of literal types. As shown in \cite{CantoneDMO21}, all the atoms in \eqref{complete MLS suite} and their negations can be rewritten in terms of  atoms of type $x \in y$ and $x = y \setminus z$ only by repeatedly applying the following equivalences much as rewrite rules (the existential quantifiers are then just dropped while the quantified variables are replaced by fresh ones):
\begin{itemize}
\item $\models x = \varnothing \:\longleftrightarrow\: x = x \setminus x$,

\item $\models x \neq \varnothing \:\longleftrightarrow\: (\exists w) w \in x$,

\item $\models x \notin y \:\longleftrightarrow\: (\exists w) (x \in w\ \wedge\ w = w \setminus y)$,

\item $\models x = y \:\longleftrightarrow\: (\exists e) (x = y \setminus e\ \wedge\ e = e\setminus e)$,

\item $\models x = y \cap z \:\longleftrightarrow\: (\exists w)(w = y \setminus z\ \wedge\ x = y \setminus w)$,

\item $\models x = y \cup z \:\longleftrightarrow\: (\exists e,w)(w = x \setminus y\ \wedge\ w = z \setminus y\ \wedge\ e = y \setminus x\ \wedge\ e = e \setminus e)$,

\item $\models x \subseteq y \:\longleftrightarrow\: x = y \cap x$,

\item $\models x \neq y \:\longleftrightarrow\: (\exists v,w,z)(w = x \cup y\ \wedge\ z = x \cap y\ \wedge\ v \in w\ \wedge\ v \notin z)$,

\item $\models x \neq y \star z \:\longleftrightarrow\: (\exists w)(x \neq w\ \wedge\ w = y \star z)$,
\end{itemize}
where $\star \in \{\cup,\cap,\setminus\}$.

Henceforth, we will restrict ourselves to \MLS-formulae that are conjunctions of atoms of the following two types only:
\begin{equation}\label{basic literal types}
x \in y, \quad x = y \setminus z.
\end{equation}
In the rest of the paper, these will be simply referred to  as \MLS-conjunctions.

Finally, as a piece of notation, for any given \emph{finite} set $\mathcal{L}$ of literals, we write $\bigwedge \mathcal{L}$ (resp., $\bigvee \mathcal{L}$) to denote the conjunction (resp., disjunction) of all the literals in $\mathcal{L}$.

\section{Convexity of \MLS}\label{SEC:convex}

Our main goal is to prove that the theory \MLS is convex, namely that, for any  \MLS-conjunction $\varphi$ and any given finite \emph{nonempty} set $\mathcal{E}$ of equalities among variables, we have:
\[
\models \varphi \longrightarrow \bigvee \mathcal{E} \quad \implies \quad \models \varphi \longrightarrow x = y \text{, ~for some equality $x = y$ in $\mathcal{E}$}.
\]

To prove that the theory \MLS is convex, we will proceed by way of contradiction. 

Thus, let us suppose that there exists an \MLS-conjunction $\varphi$ (namely a conjunction of literals of  type \eqref{basic literal types}) and a finite, nonempty set $\mathcal{E}$ of equalities among variables such that:

\begin{enumerate}[label=(C\arabic*)]
\item\label{C1} $\models \varphi \longrightarrow \bigvee \mathcal{E}$;

\item\label{C2} $\not\models \varphi \longrightarrow x = y$, for any $x=y$ in $\mathcal{E}$\\
(that is, for every $x=y$ in $\mathcal{E}$ there exists some set assignment $M_{x,y}$ such that $M_{x,y} \models \varphi \wedge x \neq y$).
\end{enumerate}
It is not restrictive to additionally assume that $\Vars{\mathcal{E}} \subseteq \Vars{\varphi}$.\footnote{Indeed, without disrupting conditions \ref{C1} and \ref{C2}, for any variable $x \in \Vars{\mathcal{E}}$ one may add to $\varphi$ the literal $x \in w$, where $w$ stands for some fresh variable.}%

In view of condition \ref{C2}, our conjunction $\varphi$ is satisfiable. Among all the models for $\varphi$, we select one, say $M$, that satisfies as few as possible  equalities in $\mathcal{E}$, namely such that the cardinality of $\mathcal{E}^{+}_{M} \defAs \{ \ell \in \mathcal{E} \st M \models \ell \}$ is \emph{minimal}. We also set $\mathcal{E}^{-}_{M} \defAs \{ \neg\ell \st \ell \in \mathcal{E} \setminus  \mathcal{E}^{+}_{M}\}$, so $\mathcal{E}^{-}_{M}$ is the collection of the inequalities $x \neq y$ such that $x=y$ is in $\mathcal{E}$ and $M \not\models x = y$ (hence, $M \models x \neq y$).

Plainly, we have $M \models \varphi\ \wedge\ \bigwedge \mathcal{E}^{+}_{M}\ \wedge\ \bigwedge \mathcal{E}^{-}_{M}$. Notice that, while $\bigwedge \mathcal{E}^{-}_{M}$ may be empty, the conjunction $\bigwedge \mathcal{E}^{+}_{M}$ must contain at least one literal, since  $M \models \bigvee \mathcal{E}$ by condition \ref{C1}.

Let $\overline{\ell}$ be any equality $\overline{x} = \overline{y}$ in $\bigwedge \mathcal{E}^{+}_{M}$, which will be referred to in the rest of our proof as the \emph{designated equality of} $\mathcal{E}$. We will prove that the conjunction
\[
\varphi^{*}\ \defAsEq\ \varphi\ \wedge\ \bigwedge \big(\mathcal{E}^{+}_{M} \setminus \{\overline{\ell}\} \big) \ \wedge\ \bigwedge \mathcal{E}^{-}_{M} \ \wedge\ \overline{x} \neq \overline{y}
\]
is satisfiable, thereby contradicting the assumed minimality of $M$, since for every model $M^{*}$ for $\varphi^{*}$ we would have $\mathcal{E}^{+}_{M^{*}} = \mathcal{E}^{+}_{M} \setminus \{\overline{\ell}\}$, and therefore $|\mathcal{E}^{+}_{M^{*}}| < |\mathcal{E}^{+}_{M}|$.

\medskip

Before diving into the details of the proof, we provide an overview of how the set assignment $M$ can be suitably enlarged into another set assignment $M^{*}$ that satisfies all the conjuncts of $\varphi \wedge \bigwedge \mathcal{E}^{+}_{M} \wedge \bigwedge \mathcal{E}^{-}_{M}$ but the designated equality $\overline{\ell}$, thus proving that $\varphi^{*}$ is satisfiable.

\newcommand{\myset}{\mathfrak{s}}
\newcommand{\myt}{\mathfrak{t}}

\subsubsection*{Proof overview}
The construction of $M^{*}$ consists in two phases: the first one, the \textsc{Boolean phase}, takes care of the satisfiability of the Boolean literals of $\varphi^{*}$, namely the literals in $\varphi^{*}$ of type $x = y \setminus z$, $x = y$, and $x \neq y$, whereas the second one, the \textsc{membership phase}, takes care of the satisfiability of the membership literals of $\varphi^{*}$, namely those of the form $x \in y$.

In order to model $\overline{x} \neq \overline{y}$, we add to exactly one between $M\overline{x}$ and $M\overline{y}$ a new member $\myset$ not already occurring in $\bigcup_{x \in \Vars{\varphi}} Mx$. The set $\myset$ must be chosen with care to prevent that no set produced during the subsequent membership phase is new to the current set assignment. In addition, the set $\myset$ must be added to the right sets $Mx$ in order that the resulting assignment keeps satisfying all of the Boolean literals in $\varphi \wedge \bigwedge \mathcal{E}^{+}_{M} \wedge \bigwedge \mathcal{E}^{-}_{M}$ other than the designated equality $\overline{x} = \overline{y}$. The first problem is solved by selecting as $\myset$ any set of rank strictly greater than that of $M$. As for the second condition, recalling that, by  \ref{C2}, the conjunction $\varphi \wedge \overline{x} \neq \overline{y}$ is satisfiable, we can select a model $\overline{M}$ for it. Therefore $\overline{M}\overline{x} \neq \overline{M}\overline{y}$, and so we can pick some element $\myt$ belonging to exactly one of the sets $\overline{M}\overline{x}$ and $\overline{M}\overline{y}$. By adding our special set $\myset$ as an element to all and only those sets $Mx$ such that $\myt \in \overline{M}x$, for $x \in \Vars{\varphi}$, we obtain a new assignment, which will be denoted $M_{0}$. It turns out that $M_{0}$ correctly models all the conjuncts in $\varphi^{*}$, but the membership literals $x \in y$ for which $M_{0}x \neq Mx$. We denote by $\mathsf{V}_{0}$ the collection of variables $x$ in $\varphi$ such that $M_{0}x \neq Mx$.

\begin{example}
We illustrate the Boolean phase of our enlargement process with the following \MLS-conjunction 
\[
\varphi\ \defAsEq \  x = \overline{y} \setminus z\ \wedge\ x = \overline{x} \setminus w\ \wedge\ x \neq \overline{y}\ \wedge\
\overline{y} \in w\ \wedge\ w \in v\ \wedge\ z \in v
\]
and with the equality $\overline{x} = \overline{y}$.

Let $M$ and $\overline{M}$ be the set assignments over $\Vars{\varphi} = \{v, w, x, \overline{x}, \overline{y}, z\}$ so defined, where to enhance readability we use the shorthand $\{\emptyset\}^{2} \defAs \{\{\emptyset\}\}$---likewise, $\{\emptyset\}^{4}$ will denote the set $\{\{\{\{ \emptyset \}\}\}\}$:
\begin{align*}
 Mx &= \emptyset,& M\overline{x} &= M\overline{y} = \{\emptyset\},& Mz &= Mw = \{\emptyset,\{\emptyset\}\},& Mv &= \{\{\emptyset,\{\emptyset\}\}\},\\
 \overline{M}x &= \emptyset,& \overline{M}\overline{y} &= \overline{M}z = \{\emptyset\},& 
\overline{M}\overline{x} &= \overline{M}w = \{\emptyset\}^{2},& \overline{M}v &= \{\{\emptyset\},\{\emptyset\}^{2}\}.
\end{align*}

It can easily be checked that $M \models \varphi \wedge \overline{x} = \overline{y}$ ~and~ $\overline{M} \models \varphi \wedge \overline{x} \neq \overline{y}$ hold.

Let $\myset \defAs \{\emptyset\}^{4}$, so that $\rk{\myset} = 4 > 3 = \rk{Mv} = \rk{M}$. Since $\emptyset \in \overline{M}\overline{y} \setminus \overline{M}\overline{x}$, we can put $\myt \defAs \emptyset$, and so we have:
\begin{align*}
M_{0}u = \begin{cases}
\{\emptyset,\myset\}& \text{if } u = \overline{y} \\
\{\emptyset, \{\emptyset\},\myset\}& \text{if } u = z\\
Mu& \text{otherwise}
\end{cases}
\end{align*}
and $\mathsf{V}_{0} = \{\overline{y}, z\}$.

Plainly, $M_{0}$ satisfies all literals in $\varphi \wedge \overline{x} \neq \overline{y}$ but the literals $\overline{y} \in w$ and $z \in v$. \qed
\end{example}

The subsequent membership phase performs the following enlargement step, for $k=0,1,2,\ldots$, until needed:
\begin{quote}
extend the assignment $M_{k}$ by putting, for each $u \in \mathsf{V}_{k}$,
\[
M_{k+1}u\ \defAs\ M_{k}u\ \cup\ \{ M_{k}v \st v \in \mathsf{V}_{k} \text{ and } Mv \in Mu\},
\]
while setting $M_{k+1}u \defAs M_{k}u$ for the remaining variables $u$ in $\Vars{\varphi}$, and define $\mathsf{V}_{k+1}$ as the collection of variables $u$ in $\Vars{\varphi}$ such that $M_{k+1}u \neq M_{k}u$.
\end{quote}
For $k = 0,1,2,\ldots$, it turns out that each $M_{k}$ correctly models all the Boolean literals in $\varphi^{*}$ and all the membership literals in $\varphi^{*}$ but those of the form $x \in y$ with $x \in \mathsf{V}_{k+1}$. Hence, as soon as some $\mathsf{V}_{k}$ is empty, the assignment $M_{k}$ is plainly a model for $\varphi^{*}$, and so the membership phase can stop. By the well-foundedness of the membership relation, such a situation occurs in at most $\overline{n} \defAs |\Vars{\varphi}|$ steps, and therefore $M_{\overline{n}}$ is a model for $\varphi^{*}$, proving that $\varphi^{*}$ is satisfiable.

\setcounter{theorem}{0}\begin{example}[cont'd]
We continue our example by illustrating the membership phase of our enlargement process. We recall that $\mathsf{V}_{0} = \{\overline{y}, z\}$. Since $M\overline{y} \in Mw = Mz$ and $Mz \in Mv$, we have $\mathsf{V}_{1} = \{z,w,v\}$, $M_{1}z = \{\emptyset,\{\emptyset\},\mathfrak{s}, \{\emptyset,\mathfrak{s}\}\}$, $M_{1}w = \{\emptyset,\{\emptyset\}, \{\emptyset,\mathfrak{s}\}\}$, $M_{1}v = \{ \{\emptyset,\{\emptyset\}\}, \{\emptyset,\{\emptyset\},\mathfrak{s}\}\}$, and $M_{1}u = M_{0}u$ for all $u \neq z,w,v$. Next, since $Mw \in Mv$ and $Mz \in Mv$, we have $\mathsf{V}_{2} = \{v\}$, $M_{2}u = M_{1}u$ for all $u \neq v$, and  \\ 
\centerline{$
M_{2}v = \{ \{\emptyset,\{\emptyset\}\}, \{\emptyset,\{\emptyset\},\mathfrak{s}\}, \{\emptyset,\{\emptyset\},\mathfrak{s},\{\emptyset,\mathfrak{s}\}\},
 \{\emptyset,\{\emptyset\}, \{\emptyset,\mathfrak{s}\}\}\}.
$}
Finally, since $Mv \notin \bigcup_{u \in \Vars{\varphi}} Mu$, we can actually stop. In fact, at this point we have 
\[
M_{2} = M_{3} = M_{4} = \cdots.
\] 
Plainly, $M_{2} \models \varphi \wedge \overline{x} \neq \overline{y}$. 	\qed
\end{example}

\subsubsection*{Proof details}

For any $V \subseteq \Vars{\varphi}$, we will use the notation $MV$ to denote the set $\{Mv \st v \in V\}$. Let $\myset$ be any fixed set whose rank is larger than the rank of $M$, namely such that $\rk{\myset} > \rk{M}$. 

We define by recursion two sequences $\{\mathsf{V}_{n}\}_{n \in \nats}$ and $\{M_{n}\}_{n \in \nats}$, respectively of subsets of $\Vars{\varphi}$ and of set assignments over $\Vars{\varphi}$, by putting:\begin{align}
\mathsf{V}_0 &\coloneqq \{u \in \Varsf\ |\ \myt \in \overline{M}u\}, \label{DEF:v0}\\
\mathsf{V}_n &\coloneqq \{u \in \Varsf\ |\ Mu \cap M \mathsf{V}_{n-1} \neq \emptyset \}, ~~\text{for } n \geq 1, \label{DEF:vn}
\shortintertext{and}
M_{0}v &\coloneqq
\begin{cases}
Mv \cup \{\myset\} &\text{if } v \in \mathsf{V}_0\\
Mv &\text{if } v \in \Vars{\varphi} \setminus \mathsf{V}_0,
\end{cases}\label{DEF:m0}\\
M_{n} v &\coloneqq \begin{cases}
M_{n-1}v \cup M_{n-1} \{u \in \mathsf{V}_{n-1} \st Mu \in Mv\} & \text{if } v \in \mathsf{V}_n\\
M_{n-1} v & \text{if } v \in \Vars{\varphi} \setminus \mathsf{V}_n,
\end{cases} \label{DEF:mn}\\
&\phantom{\coloneqq~~~}\text{for $n \geq 1$ and $v \in \Vars{\varphi}$.}\notag
\end{align}

As a direct consequence of \eqref{DEF:m0} and \eqref{DEF:mn}, the following results can be easily proved by induction:
\begin{lemma}\label{LEMMA:subsetmodel}
\begin{enumerate}[label=(\alph*)]
\item\label{LEMMA:subsetmodelA}
For every $v \in \Varsf$, we have
\[
Mv\ \subseteq\ M_{0}v\ \subseteq\ \cdots\ \subseteq\ M_{n} v\ \subseteq\ \cdots\/.
\]

\item\label{LEMMA:subsetmodelB} For all $v \in \Vars{\varphi}$ and $n \in \nats$, we have:
\[
M_{n}v\ \subseteq\ M v\ \cup\ \{\myset\}\ \cup\ \bigcup_{k=0}^{n-1} M_{k} \{u \in \mathsf{V}_{k} \st Mu \in Mv\}.
\]
\end{enumerate}
\end{lemma}

Lemma~\ref{LEMMA:subsetmodel}\ref{LEMMA:subsetmodelA} implies that the sequence of assignments $\{M_{n}\}_{n \in \nats}$ is plainly pointwise convergent. As a consequence of the next lemma and corollary, it will follow in fact that $\{M_{n}\}_{n \in \nats}$ converges ``uniformly'', and it does so in at most $|\Vars{\varphi}|$ steps.

\begin{lemma}\label{LEMMA:end}
Let $k \in \nats$. We have:
\begin{enumerate}[label=(\alph*)]
\item\label{LEMMA:endA} if $\mathsf{V}_{k} = \emptyset$ then, for all $n \geq k $,
\begin{enumerate}[label=(a$_{\arabic*}$)]
\item $\mathsf{V}_{n} = \emptyset$, 

\item $M_{n} = M_{k}$;
\end{enumerate}
\item\label{LEMMA:endB} if $\mathsf{V}_{k} \neq \emptyset$, then
\begin{equation}\label{eq:LEMMA:endB}
k\ \leq\ \min \big( |\Vars{\varphi}|-1, \rk{M} \big).
\end{equation}
\end{enumerate}
\end{lemma}
\begin{proof}
If $\mathsf{V}_k = \emptyset$, then $\mathsf{V}_{k+1} = \emptyset$ and $M_{k+1} = M_{k}$ by \eqref{DEF:vn} and \eqref{DEF:mn}, respectively. By iterating the same argument, one can easily prove that $\mathsf{V}_{n} = \emptyset$ and $M_{n} = M_{k}$, for all $n \in \nats$, proving \ref{LEMMA:endA}.

As for \ref{LEMMA:endB}, we preliminarily observe that, by \eqref{DEF:vn}, for all $v \in \Vars{\varphi}$ and $n \geq 1$ we have 
\begin{equation}\label{v in Vn}
v \in \mathsf{V}_{n} \: \implies \: (\exists u \in \mathsf{V}_{n-1}) Mu \in Mv.
\end{equation}
Thus, if  $\mathsf{V}_{k} \neq \emptyset$, by picking any $v_{k} \in \mathsf{V}_{k}$ and by repeatedly applying \eqref{v in Vn}, it follows that there exist $v_{0},v_{1},\ldots,v_{k-1} \in \Vars{\varphi}$ such that
\begin{equation}\label{chain of Mv's}
Mv_{0}\ \in\ Mv_{1}\ \in\ \cdots\ \in\ Mv_{k-1}\ \in\ Mv_{k}.
\end{equation}
By the well-foundedness of $\in$, the variables $v_{0},v_{1},\ldots,v_{k-1},v_{k}$ must be pairwise distinct. Hence, $k+1 \leq |\Vars{\varphi}|$. In addition, \eqref{chain of Mv's} also yields $k \leq \rk{Mv_{k}} \leq \rk{M}$. Thus, \eqref{eq:LEMMA:endB} follows, proving \ref{LEMMA:endB}.
\end{proof}

The preceding lemma yields immediately the following result.
\begin{corollary}\label{co:wasLEMMA:endC}
For all $h,k > \min \big( |\Vars{\varphi}|-1, \rk{M} \big)$, we have $M_{h} = M_{k}$.
\end{corollary}

Letting $\overline{n} \coloneqq |\Vars{\varphi}|$, Corollary~\ref{co:wasLEMMA:endC} implies that $M_{n} = M_{\overline{n}}$, for all $n \geq \overline{n}$.

\smallskip

Next we prove a number of technical lemmas that will culminate in the proof that
\[
M_{\overline{n}}\ \models\ \varphi\ \wedge\ \bigwedge \mathcal{E}^{-}_{M}\ \wedge\ \overline{x} \neq \overline{y},
\]
where $\overline{x} = \overline{y}$ is the designated equality of $\mathcal{E}$. Thus, we will have that
\[
\mathcal{E}^{-}_{M} \subsetneq \mathcal{E}^{-}_{M_{\overline{n}}} \quad \text{and} \quad |\mathcal{E}^{+}_{M_{\overline{n}}}| < |\mathcal{E}^{+}_{M}|,
\]
contradicting the minimality of $|\mathcal{E}^{+}_{M}|$. Hence, the convexity of \MLS will follow, since our initial assumption on $\varphi$ and $\mathcal{E}$ that conditions \ref{C1} and \ref{C2} hold will be proved to be untenable.

The following lemma provides some useful bounds on the rank of $M_{n}v$, for $n \in \nats$ and $v \in \Vars{\varphi}$.
\begin{lemma}\label{LEMMA:rank}
For all $n \in \nats$ and $v \in \Vars{\varphi}$, we have
\begin{enumerate}[label=-]
\item $\rk{M_{n}v} = \rk{\myset} + n + 1$, if $v \in \mathsf{V}_n$,

\item $\rk{M_{n}v} \leq \rk{\myset} + n$, if $v \notin \mathsf{V}_n$.
\end{enumerate}
\end{lemma}
\begin{proof}
We proceed by induction on $n$. For $n = 0$ and $v \in \mathsf{V}_0$, from \eqref{DEF:m0} we have $M_{0}v = Mv\ \cup\ \{\myset\}$. Hence, $\rk{M_{0}v} = \max\{\rk{Mv},\rk{\{\myset\}}\} = \rk{\myset}+1$, since $\rk{\myset} > \rk{Mv}$. On the other hand, if $v \notin \mathsf{V}_0$, then $\rk{M_{0}v} = \rk{Mv} < \rk{\myset}$.

Next, let $n > 0$ and $v \in \mathsf{V}_n$. By \eqref{DEF:mn}, we have:
\begin{equation}\label{rk Mn v}
\rk{M_{n}v} = \max \big( \rk{M_{n-1}v}, \rk{M_{n-1} \{z \in \mathsf{V}_{n-1} \st Mz \in Mv\}} \big).
\end{equation}
By inductive hypothesis, we readily have
\begin{enumerate}[label=-]
\item $\rk{M_{n-1}v} \leq \rk{\myset} + n$, and

\item $\rk{M_{n-1} \{z \in \mathsf{V}_{n-1} \st Mz \in Mv\}} \leq \rk{\myset} + n + 1$.
\end{enumerate}
In addition, since $v \in \mathsf{V}_{n}$, then by \eqref{DEF:vn}, $Mu \in Mv$ for some $u \in \mathsf{V}_{n-1}$. Hence, again by inductive hypothesis, $\rk{Mu} = \rk{\myset} +n$, and since $u \in \{z \in \mathsf{V}_{n-1} \st Mz \in Mv\}$, we have 
\[
\rk{M_{n-1} \{z \in \mathsf{V}_{n-1} \st Mz \in Mv\}} = \rk{\myset} + n + 1.
\]
Thus, by \eqref{rk Mn v}, we get $\rk{M_{n}v} = \rk{\myset} + n + 1$.

On the other hand, if $v \notin \mathsf{V}_n$, then by \eqref{DEF:mn} and by the inductive hypothesis we have $\rk{M_{n}v} = \rk{M_{n-1}v} \leq \rk{\myset} + n$.
\end{proof}

Next we prove that the set $\myset$ can enter $M_{n}$ only when $n=0$.
\begin{lemma}\label{merged lemmas}
For all $n \in \nats$ and $v \in \Varsf$, we have:
\begin{enumerate}[label=(\alph*)]\item\label{new lemma} $M_{n}v \neq \myset$;

\item\label{LEMMA:myset} $\myset \in M_{n}x \quad \Longleftrightarrow \quad \myset \in M_{0}x$.
\end{enumerate}
\end{lemma}
\begin{proof}Concerning \ref{new lemma}, we proceed by induction on $n$.

For $n = 0$, by \eqref{DEF:m0} we have:
\begin{enumerate}[label=-]
\item $\rk{M_{0}v} = \rk{\myset} +1$, if $v \in \mathsf{V}_{0}$ \quad (by Lemma~\ref{LEMMA:rank});

\item $\rk{M_{0}v} = \rk{Mv} < \rk{\myset}$, if $v \notin \mathsf{V}_{0}$.
\end{enumerate}
In both cases, it follows that $M_{0}v \neq \myset$.

For the inductive step, let $n > 1$. If $v \in \mathsf{V}_{n}$, then by Lemma~\ref{LEMMA:rank} we have $\rk{M_{n}v} = \rk{\myset} + n + 1$, and therefore $M_{n}v \neq \myset$. On the other hand, if $v \notin \mathsf{V}_{n}$, then $M_{n}v = M_{n-1}v \neq \myset$, by \eqref{DEF:mn} and by the inductive hypothesis.

\smallskip

Next we prove \ref{LEMMA:myset} by induction on $n$.

The base case $n=0$ is trivial.

For the inductive step, let $n > 0$. If $\myset \in M_{0}x$, then Lemma~\ref{LEMMA:subsetmodel}\ref{LEMMA:subsetmodelA} yields readily $\myset \in M_{n}x$. Conversely, let $\myset \in M_{n}x$. If $x \notin \mathsf{V}_{n}$, then by \eqref{DEF:mn} we have $\myset \in M_{n}x = M_{n-1}x$, and therefore by inductive hypothesis $s \in M_{0}x$. On the other hand, if $x \in \mathsf{V}_{n}$, then again by \eqref{DEF:mn} we have
\[
\myset\ \in\ M_{n}x\ =\ M_{n-1}x\ \cup\ M_{n-1} \{y \in \mathsf{V}_{n-1} \st My \in Mx\}.
\]
In view of \ref{new lemma}, the latter formula yields $\myset \in M_{n-1}x$, and therefore $s \in M_{0}x$ follows again by inductive hypothesis, completing the proof of \ref{LEMMA:myset}, and in turn of the lemma.
\end{proof}

\newcommand{\myq}{\mathfrak{q}}The following lemma proves that, at each construction step of the assignments $M_{n}$'s, only elements of rank at least $\rk{\myset}$ can enter into play.
\begin{lemma}\label{preCorollary:noteq}
For every set $\myq \in M_{n}x$, for some $x \in \Vars{\varphi}$ and $n \in \nats$,  if $\rk{\myq} < \rk{\myset}$ then $\myq \in Mx$.
\end{lemma}
\begin{proof}
Let $x \in \Vars{\varphi}$, $n \in \nats$, and $\myq \in M_{n}x$, with $\rk{\myq} < \rk{\myset}$.  From Lemma~\ref{LEMMA:subsetmodel}\ref{LEMMA:subsetmodelB}, we have 
\[
M_{n}x\ \subseteq\ M x\ \cup\ \{\myset\}\ \cup\ \bigcup_{k=0}^{n-1} M_{k} \{y \in \mathsf{V}_{k} \st My \in Mx\}.
\]
Since, by Lemma~\ref{LEMMA:rank}, the rank of each member of $\{\myset\} \cup \bigcup_{k=0}^{n-1} M_{k} \{y \in \mathsf{V}_{k} \st My \in Mx\}$ is greater than or equal to $\rk{\myset}$, then necessarily $\myq \in Mx$
\end{proof}

All the inequalities $x \neq y$ satisfied by $M$ are satisfied by every $M_{n}$, as proved in the following corollary.
\begin{corollary}\label{Corollary:noteq}
If $Mx \neq My$, for some $x, y \in \Varsf$, then $M_{n}x \neq M_{n}y$, for every $n \in \nats$.
\end{corollary}
\begin{proof}
W.l.o.g., let us assume that $Mx \nsubseteq My$, and let $\myq \in Mx \setminus My$. Also, let $n \in \nats$. By Lemma~\ref{LEMMA:subsetmodel}\ref{LEMMA:subsetmodelA}, $\myq \in M_{n}x$. Plainly, $\rk{\myq} < \rk{Mx} < \rk{\myset}$ and $\myq \notin My$. Thus, Lemma~\ref{preCorollary:noteq} yields $\myq \notin M_{n}y$, proving that $M_{n}x \neq M_{n}y$.
\end{proof}

To show that every membership $x \in y$ satisfied by $M$ is correctly modeled by $M_{\overline{n}}$, we will need the following result.
\begin{lemma}\label{LEMMA:A}
For all $n \in \nats$ and $x,y \in \Vars{\varphi}$, if $x \in \mathsf{V}_{n}$ and $Mx \in My$, then $y \in \mathsf{V}_{n+1}$ and $M_{n} x \in M_{n+1} y$.
\end{lemma}
\begin{proof}
Let $n \in \nats$ and assume that $Mx \in My$, for some $x,y \in \Vars{\varphi}$, and that $x \in \mathsf{V}_{n}$. Then, by \eqref{DEF:vn}, $y \in \mathsf{V}_{n+1}$. In addition, from \eqref{DEF:mn}, the latter membership relation yields immediately that $M_{n} x \in M_{n+1} y$.
\end{proof}

We are now ready to prove our main lemma.
\begin{lemma}\label{LEMMA:model}
The assignment $M_{\overline{n}}$ satisfies $\varphi$.
\end{lemma}
\begin{proof}
We prove the lemma, by showing that $M_{\overline{n}}$ correctly models all the conjuncts in $\varphi$. We recall that, in view of the reduction process outlined in Section~\ref{SEC:MLSsemantics}, our formula $\varphi$ contains conjuncts of two types only, namely $x \in y$ and $x = y \setminus z$.

\paragraph{Conjuncts of type $x \in y$.} Let $x \in y$ occur in $\varphi$, so that $Mx \in My$ holds. If $x \notin \mathsf{V}_{n}$ for all $n \in \nats$, then $M_{\overline{n}} x = M x \in M y \subseteq M_{\overline{n}}y$ (by Lemma~\ref{LEMMA:subsetmodel}\ref{LEMMA:subsetmodelA}), from which $M_{\overline{n}} x \in M_{\overline{n}}y$ follows.\\
Conversely, if $x \in \mathsf{V}_{n}$, for some $n \in \nats$, we set $\overline{m} \defAs \max \{n \in \nats \st x \in \mathsf{V}_{n}\}$. In addition, since $Mx \in My$ and $x \in \mathsf{V}_{\overline{m}}$, Lemma~\ref{LEMMA:A} implies $y \in \mathsf{V}_{\overline{m}+1}$ and therefore, by Lemma~\ref{LEMMA:end}\ref{LEMMA:endB}, $\overline{m}+1 \leq |\Vars{\varphi}| - 1 \leq \overline{n}$. Thus, Lemma~\ref{LEMMA:A} again together with Lemma~\ref{LEMMA:subsetmodel}\ref{LEMMA:subsetmodelA} yields $M_{\overline{n}} x = M_{\overline{m}} x \in M_{\overline{m}+1} y \subseteq M_{\overline{n}} y$, from which $M_{\overline{n}}x \in M_{\overline{n}} y$ follows.

\paragraph{Conjuncts of type $x = y \setminus z$.} Let $x = y \setminus z$ occur in $\varphi$, so that $Mx = My \setminus Mz$ holds. We will prove that $M_{\overline{n}} \models x = y \setminus z$, by proving that $M_{\overline{n}}x \subseteq M_{\overline{n}}y \setminus M_{\overline{n}}z$ and $M_{\overline{n}}y \setminus M_{\overline{n}}z \subseteq M_{\overline{n}}x$ hold.\\[.2cm]
\emph{Proof of $M_{\overline{n}}x \subseteq M_{\overline{n}}y \setminus M_{\overline{n}}z$.} \quad From Lemma~\ref{LEMMA:subsetmodel}\ref{LEMMA:subsetmodelB}, we have:
\[
M_{\overline{n}}x\ \subseteq\ M x\ \cup\ \{\myset\}\ \cup\ \bigcup_{k=0}^{\overline{n}-1} M_{k} \{u \in \mathsf{V}_{k} \st Mu \in Mx\}.
\]
Let $\myq \in M_{\overline{n}} x$. We first consider that case in which $\myq \in M x$. Then $\myq \in M y \setminus M z$. Hence, by Lemma~\ref{LEMMA:subsetmodel}\ref{LEMMA:subsetmodelA}, $\myq \in M_{\overline{n}} y$. In addition, since $\myq \notin M z$ and $\rk{\myq} < \rk{\myset}$, Lemma~\ref{preCorollary:noteq} yields $\myq \notin M_{\overline{n}} z$. Thus, $\myq \in M_{\overline{n}} y \setminus M_{\overline{n}} z$.\\[.1cm]
Next, let $\myq = \myset$. Hence, $\myset \in M_{0} x$ (by Lemma~\ref{merged lemmas}\ref{LEMMA:myset}), so that $x \in \mathsf{V}_{0}$ (by \eqref{DEF:m0} and \eqref{DEF:v0}), and therefore $\myt \in \overline{M}x$. Since $\overline{M}x = \overline{M}y \setminus \overline{M}z$, then $\overline{M}x \subseteq \overline{M}y$, and so $\myt \in \overline{M}y$ and $\myt \notin \overline{M}z$; hence, $y \in \mathsf{V}_{0}$ and $z \notin \mathsf{V}_{0}$. Therefore $\myset \in M_{0}y \subseteq M_{\overline{n}}y$ (by \eqref{DEF:m0} and Lemma~\ref{LEMMA:subsetmodel}\ref{LEMMA:subsetmodelB}) and $\myset \notin M_{0}z$ (by \eqref{DEF:v0}). Thus, by Lemma~\ref{merged lemmas}\ref{LEMMA:myset}, $\myset \notin M_{\overline{n}}z$. In conclusion, if $\myq = \myset$ then $\myq \in M_{\overline{n}}y \setminus M_{\overline{n}}z$, as in the preceding case.\\[.1cm]
Finally, let $\myq = M_{k}u$, for some $0 \leq k < \overline{n}$ and $u \in \mathsf{V}_{k}$ such that $Mu \in Mx$. Recalling that $Mx = My \setminus Mz$, then $Mu \in My$, so that $M_{k}u \in M_{k+1}y \subseteq M_{\overline{n}}y$ (by Lemma~\ref{LEMMA:A}). In addition, $Mu \notin Mz$. By Lemma~\ref{LEMMA:rank}, $M_{k}u \notin Mz \cup \{\myset\}$. Since, by Lemma~\ref{LEMMA:subsetmodel}\ref{LEMMA:subsetmodelB},
\begin{equation}\label{double star}
M_{\overline{n}}z\ \subseteq\ M z\ \cup\ \{\myset\}\ \cup\ \bigcup_{k=0}^{\overline{n}-1} M_{k} \{v \in \mathsf{V}_{k} \st Mv \in Mz\},
\end{equation}
to prove that $M_{k}u \notin M_{\overline{n}}z$, it is sufficient to show that $M_{k}u \notin \bigcup_{h=0}^{\overline{n}-1} M_{h} \{v \in \mathsf{V}_{h} \st Mv \in Mz\}$. By way of contradiction, let us assume that $M_{k}u = M_{h}v$, for some $0\leq h < \overline{n}$ and $v \in \mathsf{V}_{h}$ such that $Mv \in Mz$. By Lemma~\ref{LEMMA:rank}, and since $u \in \mathsf{V}_{k}$, we must have $h = k$. Since $Mu \notin Mz$ while $Mv \in Mz$, we plainly have $Mu \neq Mv$. Hence, by Corollary~\ref{Corollary:noteq}, $M_{h}u \neq M_{h}v = M_{k}u$, which contradicts our preceding assumption $M_{k}u = M_{h}v$. Thus, $M_{k}u \notin \bigcup_{h=0}^{\overline{n}-1} M_{h} \{v \in \mathsf{V}_{h} \st Mv \in Mz\}$ holds. In view of $M_{k}u \notin Mz \cup \{\myset\}$ and \eqref{double star}, the latter equation implies $M_{k}u \notin M_{\overline{n}}z$, proving that $\myq \in M_{\overline{n}}y \setminus M_{\overline{n}}z$ even in the case in which $\myq \in \bigcup_{k=0}^{\overline{n}-1} M_{k} \{u \in \mathsf{V}_{k} \st Mu \in Mz\}$.

From the arbitrariness of $\myq \in M_{\overline{n}}x$, we conclude that $M_{\overline{n}}x \subseteq M_{\overline{n}}y \setminus M_{\overline{n}}z$ holds. \\[.2cm]
\emph{Proof of $M_{\overline{n}}y \setminus M_{\overline{n}}z \subseteq M_{\overline{n}}x $.} \quad Let us assume now that $\myq \in M_{\overline{n}}y \setminus M_{\overline{n}}z$, so that $\myq \in M_{\overline{n}}y$. Again from Lemma~\ref{LEMMA:subsetmodel}\ref{LEMMA:subsetmodelB}, we have:
\[
M_{\overline{n}}y\ \subseteq\ M y\ \cup\ \{\myset\}\ \cup\ \bigcup_{k=0}^{\overline{n}-1} M_{k} \{v \in \mathsf{V}_{k} \st Mv \in My\}.
\]
First we consider the case in which $\myq \in M y$, so that $\rk{\myq} < \rk{\myset}$. Since $\myq \notin M_{\overline{n}}z$, then by Lemma~\ref{LEMMA:subsetmodel}\ref{LEMMA:subsetmodelA} $\myq \notin M z$, and therefore $\myq \in My \setminus Mz = Mx \subseteq M_{\overline{n}}x$. Thus, $\myq \in M_{\overline{n}}x$.\\[.1cm]
Next, if $\myq = \myset$, then $\myset \in M_{\overline{n}}y \setminus M_{\overline{n}}z$. Thus, $\myset \in M_{0}y$ and $\myset \notin M_{0}z$ by Lemmas~\ref{merged lemmas}\ref{LEMMA:myset} and \ref{LEMMA:subsetmodel}\ref{LEMMA:subsetmodelA}, respectively. Hence, by \eqref{DEF:m0}, $y \in \mathsf{V}_{0}$ and $z \notin \mathsf{V}_{0}$, so that $\myt \in \overline{M}y \setminus \overline{M}z = \overline{M}x$ (since $\overline{M} \models \varphi$). In view of \eqref{DEF:v0}, the latter membership relation yields $x \in \mathsf{V}_{0}$. Thus $\myq = \myset \in M_{0}x \subseteq M_{\overline{n}}x$, which readily implies $\myq \in M_{\overline{n}}x$.\\[.1cm]
Finally, let us assume that $\myq = M_{k}v$, for some $0\leq k < \overline{n}$ such that $v \in \mathsf{V}_{k}$ and $Mv \in My$. Plainly, $Mv \notin Mz$, otherwise by Lemma~\ref{LEMMA:A} we should have $\myq = M_{k}v \in M_{k+1}z \subseteq M_{\overline{n}}z$, contradicting $\myq \in M_{\overline{n}}y \setminus M_{\overline{n}}z$. Thus, $Mv \in My \setminus Mz = Mx$, so that $Mv \in Mx$. But then, by Lemma~\ref{LEMMA:A} again, we get $\myq = M_{k}v \in M_{k+1}x \subseteq M_{\overline{n}}x$, from which $\myq \in M_{\overline{n}}x$ follows even in the last case.\\[.1cm]
Thus, in all cases we have $\myq \in M_{\overline{n}}x$. By the arbitrariness if $\myq$ in $M_{\overline{n}}y \setminus M_{\overline{n}}z$, we therefore obtain $M_{\overline{n}}y \setminus M_{\overline{n}}z \subseteq M_{\overline{n}}x$.\\[.2cm]
In view of the reverse inclusion $M_{\overline{n}}x \subseteq M_{\overline{n}}y \setminus M_{\overline{n}}z$ established earlier, the latter inclusion yields $M_{\overline{n}}x = M_{\overline{n}}y \setminus M_{\overline{n}}z$, namely $M_{\overline{n}} \models x = y \setminus z$. 

\paragraph{} Summing up, we have proved that the assignment $M_{\overline{n}}$ satisfies all the conjuncts of $\varphi$, and therefore $M_{\overline{n}}$ satisfies $\varphi$.
\end{proof}

Together with Corollary~\ref{Corollary:noteq}, the preceding lemma implies
\[
M_{\overline{n}}\ \models\ \varphi\ \wedge\ \bigwedge \mathcal{E}^{-}_{M}.
\]
To find a contradiction, it only remains to prove that $M_{\overline{n}} \models \overline{x} \neq \overline{y}$, where $\overline{x} = \overline{y}$ is the designated equality of $\mathcal{E}^{+}_{M}$, which we do next.
\begin{lemma}\label{LEMMA:difference}
The assignment $M_{\overline{n}}$ models the inequality $\overline{x} \neq \overline{y}$ correctly.
\end{lemma}
\begin{proof}

We know that $\myt \in \overline{M}\overline{x} \setminus \overline{M}\overline{y}$, therefore $\overline{x} \in \mathsf{V}_0$ and $\overline{y} \notin \mathsf{V}_0$, $\myset \in M_{0}\overline{x} \setminus M_{0} \overline{y}$. Hence, from Lemma~\ref{merged lemmas}\ref{LEMMA:myset} it follows that $\myset \in M_{\overline{n}}\overline{x} \setminus M_{\overline{n}}\overline{y}$, proving that $M_{\overline{n}} \models \overline{x} \neq \overline{y}$.
\end{proof}

From Lemmas~\ref{LEMMA:model} and \ref{LEMMA:difference} and Corollary~\ref{Corollary:noteq}, we have:
\[
M_{\overline{n}}\ \models\ \varphi\ \wedge\ \bigwedge \mathcal{E}^{-}_{M}\ \wedge\ \overline{x} \neq \overline{y}.
\]
Setting $\mathcal{E}^{+}_{M_{\overline{n}}} \defAs \{ \ell \in \mathcal{E} \st M_{\overline{n}} \models \ell \}$ and $\mathcal{E}^{-}_{M_{\overline{n}}} \defAs \{ \neg\ell \st \ell \in \mathcal{E} \setminus  \mathcal{E}^{+}_{M_{\overline{n}}}\}$, we have $\mathcal{E}^{-}_{M} \subsetneq \mathcal{E}^{-}_{M_{\overline{n}}}$ and so $|\mathcal{E}^{+}_{M_{\overline{n}}}| < |\mathcal{E}^{+}_{M}|$, contradicting the minimality of $|\mathcal{E}^{+}_{M}|$ among all the set assignments that satisfy $\varphi$. Thus, our initial hypothesis that \MLS were not convex is inadmissible, and therefore we can conclude that:
\begin{theorem}
The theory $\MLS$ is convex.
\end{theorem}

We expect that the proof of convexity of \MLS can be suitably generalized to show that also the extension \MLSI of \MLS with literals of the form $x = \bigcap y$ is convex too, where $\bigcap$ is the \emph{general intersection operator}.\footnote{The decision problem for \MLSI has been solved in \cite{CanCut89b}.} We recall that the intended semantics of $x = \bigcap y$ is the following: for a given set assignment $M$, we have $M(x = \bigcap y) = \true$ if and only if $My \neq \emptyset$ and $Mx = \bigcap My$, namely $Mx = \bigcap_{\mathfrak{y} \in My} \mathfrak{y}$.
\bigskip

Several fragments of \MLS admit polynomial-time decision procedures, so they are very appealing in the context of combination of decision procedures \emph{\`a la} Nelson-Oppen. We briefly review them next.

\subsection{Polynomial fragments of \MLS}\label{SEC:relatedworks}

Convexity of \MLS is plainly inherited by all of its fragments. In \cite{CantoneDMO21} and \cite{CMO20}, we recently investigated them with the goal of spotting the polynomial ones, namely the fragments of \MLS endowed with polynomial-time satisfiability tests. Specifically, we examined all the sublanguages of the theories 
\[
\BSTbb \defAs \Theory{\Quote{\cup},\Quote{\cap},\Quote{\setminus},\,\Quote\eqnothing,\Quote{\neqnothing},\Quote{\textsf{Disj}},\Quote{\neg \textsf{Disj}},\Quote{\subseteq},\Quote{\not\subseteq},\Quote{=},\Quote{\neq}}
\qquad \text{and} \qquad 
\MSTbb  \defAs \MSTtheory{\Quote{\cup},\Quote{\cap},\Quote{\setminus},\in,\notin},
\]
where:
\begin{enumerate}[label=-]
\item \BSTbb (acronym for Boolean Set Theory) is the collection of all the conjunctions of literals of the types
\begin{align*}
s &\eqqnothing, &~~~& \phantom{\neg}\Disj{s}{t}, ~~~& s &\subseteq t, ~~~& s &= t, \\
s &\neqqnothing, &&  \neg\Disj{s}{t}, & s &\not\subseteq t, &s &\neq t,
\end{align*}
with $s$ and $t$ terms involving set variables and the Boolean operators $\Quote{\cup}$, $\Quote{\cap}$, and $\Quote{\setminus}\:$, and where $\Disj{s}{t}$ stands for $s \cap t = \emptyset$;

\item \MSTbb (acronym for Membership Set Theory) is the collection of all the conjunctions of literals of the two types $s \in t$ and $s \notin t$, with, as above, $s$ and $t$  terms involving set variables and the Boolean operators $\Quote{\cup}$, $\Quote{\cap}$, and $\Quote{\setminus}\:$.
\end{enumerate}

More generally, we denote by $\Theory{\mathsf{op}_{1},\ldots,\,\mathsf{pred}_{1},\ldots}$ (resp., $\MSTtheory{\mathsf{op}_{1},\ldots,\,\mathsf{pred}_{1},\ldots}$) the subtheory of \BSTbb (resp., \MSTbb) involving only the set operators $\mathsf{op}_{1},\ldots$ drawn from the collection $\{\Quote{\cup},\Quote{\cap},\Quote{\setminus}\}$ and the predicate symbols $\mathsf{pred}_{1},\ldots$ drawn from $\{\Quote\eqnothing,\Quote{\neqnothing},\Quote{\textsf{Disj}},\Quote{\neg \textsf{Disj}},\Quote{\subseteq},\Quote{\not\subseteq},\Quote{=},\Quote{\neq}\}$ (resp., $\{\Quote\in,\Quote{\notin}\}$).

We figured out that the maximal polynomial fragments of \BSTbb and \MSTbb (namely the polynomial fragments of \BSTbb and \MSTbb that are not strictly contained in any polynomial fragment of \BSTbb and \MSTbb, respectively) are:
\begin{itemize}
\item $\BSTTh{\cup,\eqnothing,\neqnothing,\disj,\neg \disj,\nsubseteq,\neq}$,
\item $\BSTTh{\cup,\eqnothing,\neqnothing,\neg \disj,\subseteq,\nsubseteq,=,\neq}$,
\item $\BSTTh{\cap,\eqnothing,\neqnothing,\disj, \neg \disj,\subseteq,\nsubseteq,=,\neq}$\\
(all of which admitting a cubic-time satisfiability test) and

\item $\MSTTh{\cup,\in,\notin}$ (admitting a linear-time satisfiability test),

\item $\MSTTh{\cap,\in,\notin}$ (admitting a quadratic-time satisfiability test).
\end{itemize}

In addition, we further spotted the following non-maximal polynomial fragments of \BSTbb admitting sub-cubic satisfiability tests:
\begin{itemize}
\item $\BSTTh{\cup,\eqnothing,\neqnothing,\disj,\nsubseteq,\neq}$ (admitting a linear-time satisfiability test),

\item $\BSTTh{\cup,\disj,\neg \disj}$ (admitting a quadratic-time satisfiability test),

\item $\BSTTh{\cap,\eqnothing,\neqnothing,\disj, \neg \disj,\neq}$ (admitting a quadratic-time satisfiability test).
\end{itemize}

As already observed, all of the above fragments plainly inherit convexity from \MLS, so that, their decision procedures can be efficiently combined with the decision procedures of other convex theories with disjoint signatures within a Nelson-Oppen framework.

\smallskip

In the following section, we review various non-convex extensions of \MLS.

\section{Non-convex extensions of \MLS}\label{non-convexity}

To prove that some extensions of \MLS are non-convex, we rely on the following property. 

\begin{lemma}\label{lemma:non-convexity}
Let $\mathsf{T}$ be any extension of \MLS containing a conjunction $\varphi$ with a designated variable $\overline{x}$ such that, for some integer $k \geq 2$, we have:
\begin{enumerate}[label=-]
\item $\models \varphi\ \longrightarrow\ |\overline{x}| \leq k$,

\item $\varphi\ \wedge\ |\overline{x}| = k$ is satisfiable,
\end{enumerate}
where $|\overline{x}|$ stands for the cardinality of $\overline{x}$.
Then $\mathsf{T}$ is not convex.
\end{lemma}
\begin{proof}[Proof sketch]
Given $\varphi$, $\overline{x}$, and $k$ as in the hypotheses, it is enough to set
$
\Phi \defAs~ \varphi\ \wedge\ \bigwedge_{i=1}^{k+1} x_{i} \in \overline{x},
$
where $x_{1},\ldots,x_{k+1}$ are pairwise distinct variables not occurring in $\varphi$. Then, we have:
\[
\models \Phi\ \longrightarrow\ \bigvee_{1 \leq i < j \leq k+1} x_{i} = x_{j}\, .
\]
In addition, each conjunction $\Phi\ \wedge\ x_{i} \neq x_{j}$, with $1 \leq i < j \leq k+1$, is satisfiable. Hence, none of the statements
\[
\models \Phi\ \longrightarrow\ x_{i} = x_{j}
\]
can hold, for $1 \leq i < j \leq k+1$. Thus, the theory $\mathsf{T}$ is not convex.
\end{proof}

Using Lemma~\ref{lemma:non-convexity}, we show next that the following extensions of \MLS are non-convex:
\begin{itemize}
    \item $\MLSS = \MLS\ +$ `$\{\cdot\}$': $\MLS$ extended with the singleton operator $x = \{y\}$ (see \cite{FOS80b}),
    
    \item $\MLSP = \MLS\ +$ `$\pow{\cdot}$': $\MLS$ extended with the powerset operator (see \cite{CFS85}),
    
    \item $\MLSU = \MLS\ +$ `$\bigcup{\cdot}$': $\MLS$ extended with the general union operator (see \cite{CFS87}),

        \item $\MLSC = \MLS\ +$ `$\times$': $\MLS$ extended with the Cartesian product operator.\footnote{The decision problem for \MLSC is still open.}
\end{itemize}

Concerning the theory \MLSS, let us consider the conjunction
\[
\varphi \coloneqq~~ x = \{y\}\ \wedge\ x' = \{y'\}\ \wedge\ \overline{x} = x \cup x'\,.
\]
Then, 
\begin{enumerate}[label=-]
\item for every model $M$ for $\varphi$, we have $M \overline{x} = \{My, My'\}$, so that $|M \overline{x}| \leq 2$ holds;

\item letting $\overline{M}$ be the set assignment for $\varphi$ such that 
\[
\overline{M}y = \emptyset, \qquad \overline{M}y' = \overline{M}x = \{\emptyset\}, \qquad \overline{M}x' = \{\{\emptyset\}\} , \qquad \overline{M}\overline{x} = \{\emptyset,\{\emptyset\}\},
\]
then $\overline{M}$ satisfies $\varphi$ and $|\overline{M}\overline{x}| = 2$.
\end{enumerate}
Thus, by Lemma~\ref{lemma:non-convexity}, \MLSS is non-convex. 

\medskip

Next, as for the theory \MLSP, let us consider the conjunction
\[
\varphi \coloneqq~~ x = \varnothing\ \wedge\ y = \pow{x}\ \wedge\ \overline{x} = \pow{y}\,.
\]
Then, $\varphi$ is plainly satisfiable and, for every set assignment $M$ satisfying $\varphi$, we have $M \overline{x} = \{\emptyset, \{\emptyset\}\}$, so that $|M \overline{x}| = 2$. Thus, by Lemma~\ref{lemma:non-convexity}, \MLSP is non-convex. 

\medskip

Concerning the fragment \MLSU, let us consider the conjunction
\[
\varphi \coloneqq~~ x = \varnothing\ \wedge\ \bigcup y = x\ \wedge\ \bigcup \overline{x} = y\,.
\]
Then, 
\begin{enumerate}[label=-]
\item for every set assignment $M$ satisfying $\varphi$, we have $M \overline{x} \subseteq \{\emptyset, \{\emptyset\}\}$ so that $|M \overline{x}| \leq 2$ holds;

\item letting $\overline{M}$ be the set assignment over $\Vars{\varphi}$ such that 
\[
\overline{M}x = \emptyset, \qquad My = \{\emptyset\}, \qquad \overline{M}\overline{x} = \{\emptyset,\{\emptyset\}\},
\]
we readily have that $\overline{M}$ satisfies $\varphi$ and $|\overline{M}\overline{x}| = 2$.
\end{enumerate}
Hence, by Lemma~\ref{lemma:non-convexity}, \MLSU is non-convex. 

\medskip

Since \MLSSP is an extension of non-convex theories, namely \MLSS and \MLSP, it follows immediately that \MLSSP is non-convex as well.

\medskip

Regarding the extension \MLSC of \MLS with the Cartesian product, we have
\begin{equation}\label{eq:MLSC not convex}
\models x \times y = \varnothing\ \longrightarrow\ (x=\varnothing\ \vee\ y=\varnothing).
\end{equation}
Since the two conjunctions $x \times y = \varnothing\ \wedge x \neq \varnothing$ and $x \times y = \varnothing\ \wedge y \neq \varnothing$ are clearly satisfiable, then
\[
\not\models x \times y = \varnothing\ \longrightarrow\ x=\varnothing \qquad \text{and} \qquad \not\models x \times y = \varnothing\ \longrightarrow\ y=\varnothing.
\]
Together with \eqref{eq:MLSC not convex}, the latter statements imply that \MLSC is non-convex.

By replacing in the above proof the Cartesian product $\times$ by the unordered Cartesian product $\otimes$, one can readily show that the extension \MLSuC of \MLS with the unordered Cartesian product $\otimes$ is non-convex too.

Finally, notice that the membership relator did not play any role in the above proof of non-convexity of \MLSC and \MLSuC. Therefore, by exactly the same argument as the above, one can show that the extensions \BSTC and \BSTuC of \BST with the Cartesian product and the unordered Cartesian product are non-convex.

Summarizing, we have proved:
\begin{lemma}
The theories \MLSS, \MLSP, \MLSSP, \MLSU, \MLSC, \MLSuC, \BSTC, and \BSTuC are all non-convex.
\end{lemma}

\smallskip

\section{Conclusions}\label{SEC:conclusions}

In this paper, we have shown that the fragment of Zermelo-Fraenkel set theory called Multi-Level Syllogistic is convex. We also proved that most common extensions of \MLS studied within the field of computable set theory are non-convex.  Two possible exceptions are \MLSI, namely the extension of \MLS with the general intersection operator $\bigcap$, and the extension of \MLS with a finiteness predicate and some cardinality constraints. In fact, we conjecture that both extensions are convex, and we plan to prove it in the near future.

Although the decision problem for \MLS is \NP-complete, several of its fragments are endowed with polynomial decision procedures.
Due to the fact that convexity is inherited by all the fragments of \MLS, the ones with polynomial-time decision procedures are particularly interesting in view of their integration with other convex, stably infinite decidable theories with disjoint signatures (such as the theory of lists, linear arithmetic, etc.) within a Nelson-Oppen context.

We therefore intend to continue our investigation of sublanguages of $\MLS$ that admit a polynomial satisfiability procedure, with the ultimate goal of obtaining a complete taxonomy for the decision problem for $\MLS$ subtheories.

We also plan to explore extensions to the basic Nelson-Oppen procedure that overcome the restriction of stable infiniteness and/or of signature disjointness (such as, for instance, the politeness property \cite{RRZ05}, or the Noetherian property \cite{GNZ05}) that are particularly suited for combinations of decision procedures for fragments of set theory.

Finally, we intend to generalize to decidable fragments of pure set theory some combination results present in literature (such as the ones with integers and cardinals---see \cite{MultiZarba02,Zarba02,Zarba05}), which are currently limited to \emph{flat} sets of urelements only.

\section*{Acknowledgements}
We thank Eugenio Omodeo, University of Trieste, for his insightful comments.\\
We are also grateful to the anonymous reviewers for their observations and suggestions.

\bibliographystyle{eptcs}

\end{document}